\documentclass[a4,12pt]{article}
\UseRawInputEncoding
\usepackage{bm}
\usepackage[all]{xy}
\usepackage{graphicx}
\usepackage{verbatim}
\usepackage{wrapfig}
\usepackage{makeidx}
\usepackage{amscd}
\usepackage{url}
\usepackage{comment}
\usepackage{enumerate}
\usepackage{here}
\usepackage{latexsym}
\usepackage{array}
\usepackage{booktabs}
\usepackage{multirow}
\usepackage{mathrsfs}
\usepackage{geometry}
\usepackage{fancyhdr}
\renewcommand{\title}[1]{

\begin{center} \Large \bf #1 \end{center}
}

\renewcommand{\author}[2]{
 \begin{center} #1  \vspace{3mm} \\
  #2 \\
 \end{center}
\addvspace{\baselineskip}
}

\usepackage{amssymb}
\usepackage{amsmath}

\usepackage{amsthm}
\newtheorem{theorem}{Theorem}[section]
\newtheorem{proposition}[theorem]{Proposition}

\newtheorem{lemma}[theorem]{Lemma}

\theoremstyle{definition}

\theoremstyle{remark}

\makeatletter
\@addtoreset{equation}{section}

\makeatother

\def\propagator{\includegraphics[width=0.1\textwidth]{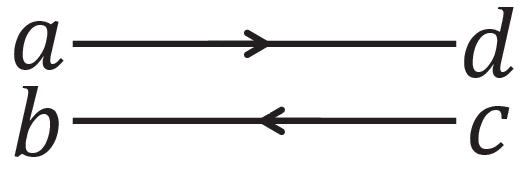}}

\def\vertexs{\includegraphics[width=0.08\textwidth]{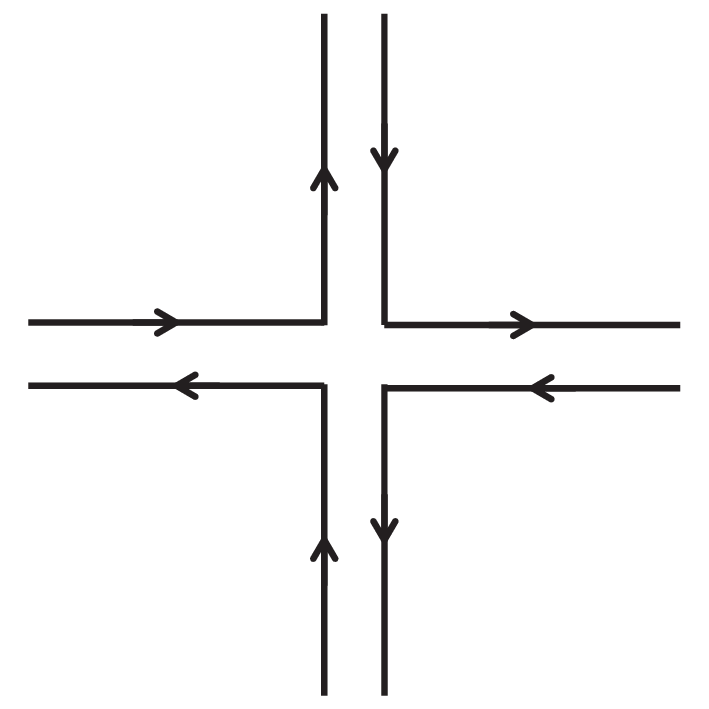}}

\begin{document}
\baselineskip 5mm
\title{Integrability of $ \Phi^4$ Matrix Model \\
 as $N$-body Harmonic Oscillator System }
\author{${}^1$Harald Grosse, ${}^2$Akifumi Sako 
}
{${}^{1,2}$
Erwin Schr\"odinger International Institute for Mathematics and Physics, \\
University of Vienna, Boltzmanngasse 9, 1090 Vienna, Austria \vspace{3mm}\\

${}^1$
Faculty of Physics, University of Vienna, Boltzmanngasse 5, 
1090 Vienna, Austria
\vspace{3mm}\\

${}^2$
Tokyo University of Science, 1-3 Kagurazaka, Shinjuku-ku, Tokyo, 162-8601, Japan
}
\noindent
\vspace{1cm}

\abstract{ 
We study a Hermitian matrix model with a kinetic
term given by $ Tr (H \Phi^2 )$, where $H$ is
a positive definite Hermitian matrix, similar as in the  Kontsevich Matrix model, but with its potential 
$\Phi^3$ replaced by $\Phi^4$.
We show that its partition function solves
an integrable Schr\"odinger-type equation 
for a non-interacting $N$-body Harmonic oscillator
system. 
}
%
%
%
\section{Introduction}\label{sect1}

Let $\Phi$ be a Hermitian $N\times N$ matrix,
$E$ be a positive diagonal $N\times N$ matrix 
$E := diag (E_1, E_2 , \cdots ,E_N )$ without
degenerate eigenvalues,
and $\eta$ be a positive real number 
as a coupling constant.
We deal in this paper with the following one Hermitian matrix model defined using this $E$:
\begin{align}
S'&= N~ Tr \{ E \Phi^2 + \frac{\eta}{4} \Phi^4  \}
\notag \\
&= N \left( 
\sum_{i,j}^N  E_{i}\Phi_{ij}\Phi_{ji}
+ \frac{\eta}{4} \sum_{i,j,k,l}^N
\Phi_{ij}\Phi_{jk}\Phi_{kl}\Phi_{li}
\right). \label{act1}
\end{align}
This matrix model is obtained 
by changing the potential of the Kontsevich model 
\cite{Kontsevich:1992ti} from $\Phi^3$ to $\Phi^4$.
It was introduced while studying a scalar field defined 
on a deformed four dimensional space-time 
and studied over years \cite{Grosse:2006tc}. 
An additional oscillator term was added 
in order to resolve the IR-UV mixing problem. 
This term leads to an external matrix 
$E$ with equally spaced Eigenvalues.
Recent developments are summarized in \cite{Branahl:2021slr}. 
\\
\bigskip

The main theorem of this paper is expressed as follows.
\begin{theorem}\label{main_thm1}
Let $Z(E, \eta)$ be the partition function defined by
$$
Z(E, \eta)= \int d \Phi ~e^{-S'} .
$$
Let $\Delta(E)$ be the Vandermonde  determinant 
$\Delta (E) := \prod_{k<l} (E_l -E_k)$.
Then the function
\begin{align}
\Psi (E, \eta ) := 
e^{-\frac{N}{2\eta} \sum_i E_i^2} \Delta(E) Z(E, \eta ) 
\label{correspPsi_Z}
\end{align}
is a zero-energy solution of a
Schr\"odinger-type differential equation being $2$-nd order 
in each of its variables,
\begin{align*}
{\mathcal H}_{HO} \Psi (E, \eta ) = 0, 
\end{align*}
where ${\mathcal H}_{HO}$ is the Hamiltonian
 of the $N$-body harmonic oscillator without interaction:
\begin{align*}
{\mathcal H}_{HO}:= - \frac{\eta}{N} \sum_{i=1}^N 
\left( 
\frac{\partial}{\partial E_i}
\right)^2  + \frac{N}{\eta}\sum_{i=1}^N (E_i)^2 .
\end{align*}
In this sense, this matrix model is a solvable system.
\end{theorem}

\section{Schwinger-Dyson Equation}\label{sect2}
Let $\Phi$ be a Hermitian $N\times N$ matrix.
Let $H$ be a positive  Hermitian $N\times N$ matrix
with nondegenerate eigenvalues
$\{E_1, E_2 , \cdots ,E_N ~ | ~ E_i \neq E_j ~\mbox{for}~ i \neq j \}$.
$\eta$ is a real positive number.
We consider the following action
\begin{align}
S&= N~ Tr \{ H \Phi^2 + \frac{\eta}{4} \Phi^4  \}
\notag \\
&= N \left( 
\sum_{i,j,k}^N  H_{ij}\Phi_{jk}\Phi_{ki}
+ \frac{\eta}{4} \sum_{i,j,k,l}^N
\Phi_{ij}\Phi_{jk}\Phi_{kl}\Phi_{li}
\right).
\label{action_S}
\end{align}
The partition function is defined by
\begin{align}
Z(E, \eta) := \int_{h_N} d \Phi ~e^{-S} ,
\label{partitionfunction}
\end{align}
and  we denote the expectation value with this action $S$
by 
$\displaystyle \langle O \rangle := \int_{h_N} d \Phi ~ O e^{-S} $.
Note that we do not normalize it here, i.e. 
$\langle 1 \rangle = Z(E, \eta) \neq 1$.
Here the integral measure is the ordinary Haar measure.
Using  the real variables defined by 
$\Phi_{ij}= \Phi_{ij}^{Re} + i \Phi_{ij}^{Im}$,
the measure is given as $\displaystyle \int_{h_N} d \Phi := 
\prod_{i}^N \int_{-\infty}^{\infty} d\Phi_{ii} 
\prod_{k<l} \int_{-\infty}^{\infty} d\Phi_{kl}^{Re}  
\int_{-\infty}^{\infty} d\Phi_{kl}^{Im}  
$.
Note that the partition function $Z(E, \eta) $ depends only on the
eigenvalues of $H$ because the integral measure is $U(N)$ invariant.
Indeed $Z(E, \eta) $ is equal to the partition function 
obtained from the action defined by $S'$ in (\ref{act1}).
\\
\bigskip

In the following, we use the notation:
\begin{align}
\frac{\partial}{\partial \Phi_{ij}}= \frac{1}{2}
\left( \frac{\partial}{\partial \Phi_{ij}^{Re}} - i 
\frac{\partial}{\partial \Phi_{ij}^{Im}}
\right) \quad \mbox{for} \ (i \neq j ) .
\end{align}
For the diagonal elements $\Phi_{ii} (i= 1,2, \cdots ,N)$, 
the corresponding partial derivatives are the
usual ones.
The Schwinger-Dyson equation is derived from 
\begin{align}
\int_{h_N} \frac{\partial}{\partial \Phi_{ij}}
\left(
\Phi_{ij} e^{-S}
\right) = 0,
\end{align}
which is expressed as
\begin{align}
Z(E, \eta) - N \sum_k (\langle
H_{ki}\Phi_{ij}\Phi_{jk}
\rangle 
+\langle
H_{jk}\Phi_{ki}\Phi_{ij}
\rangle 
)
-N\eta 
 \sum_{k,l} \langle
\Phi_{jk}\Phi_{kl}\Phi_{li}\Phi_{ij}
\rangle =0 .
\end{align}
Taking sum over the indices $i,j$ and
using
\begin{align}
\frac{\partial Z(E, \eta) }{\partial H_{ij}} =
-N \sum_k \langle  \Phi_{jk}\Phi_{ki} \rangle ,
\quad
 \frac{\partial^2 Z(E, \eta) }{\partial H_{ij}\partial H_{mn}} =
N^2  \sum_{k,l} \langle \Phi_{jk}\Phi_{ki}
\Phi_{nl}\Phi_{lm} \rangle ,
\end{align}
a partial differential equation is obtained:
\begin{align}
{\mathcal L}_{SD}^H Z(E, \eta) = 0 . \label{SD_H}
\end{align}
Here ${\mathcal L}_{SD}^H $ is a second order differential operator
defined by
\begin{align}
{\mathcal L}_{SD}^H:=
N^2 + 2 \sum_{i,k} H_{ki} \frac{\partial  }{\partial H_{ki}}
-\frac{\eta}{N} \sum_{i,k}
\left( 
\frac{\partial  }{\partial H_{ki}}\frac{\partial  }{\partial H_{ik}}
\right) .
\end{align}

Next we rewrite this Schwinger-Dyson equation by using eigenvalues
of $H$ i.e. 
$E_n (n= 1,2, \cdots , N)$.
References \cite{Itzykson:1992ya,Kimura} are helpful in the following calculations.
Let $P(x)$ be the characteristic polynomial: 
$$P(x): = \det (x~ Id_N - H) = \prod_{i=1}^N (x-E_i).$$
Using this $P(x)$,
\begin{align}
\frac{\partial E_j}{\partial H_{ki}}
= \frac{(-1)^{k+i} | E_j~ Id_N - H |_{ki} }{P'(E_j)}
\label{H-E_formula}
\end{align}
is obtained.
Here $ | M |_{kj} $ denote the minors of the matrix $M$ defined
by the determinant of the smaller matrix 
obtained by removing the $k$-th row and $j$-th column from $M$.
Using the formula (\ref{H-E_formula}),
\begin{align}
\sum_{i,j} H_{ij} \frac{\partial  Z(E, \eta)}{\partial H_{ij}}
=& \sum_{i,j,k} H_{ij} \frac{\partial  E_k}{\partial H_{ij}}
\frac{\partial  Z(E, \eta)}{\partial E_{k}} \notag \\
=& \sum_{i,j,k} 
H_{ij} \frac{(-1)^{i+j} | E_k~ Id_N - H |_{ij} }{P'(E_k)}
\frac{\partial  Z(E, \eta)}{\partial E_{k}} \notag \\
=& -\sum_{i,j,k} 
(\delta_{ij}E_k - H_{ij} ) \frac{(-1)^{i+j} | E_k~ Id_N - H |_{ij} }{P'(E_k)}
\frac{\partial  Z(E, \eta)}{\partial E_{k}}  \notag\\
&+ \sum_{i,j,k} 
\delta_{ij}E_k \frac{(-1)^{i+j} | E_k~ Id_N - H |_{ij} }{P'(E_k)}
\frac{\partial  Z(E, \eta)}{\partial E_{k}}.
\end{align}
The first term in the last line is equal to 
$\displaystyle -\sum_{k,i} \frac{P(E_k)}{P'(E_k)} \frac{\partial  Z(E, \eta)}{\partial E_{k}} = 0$,
because $P(x)$ is the characteristic polynomial and $E_k$ is
one of eigenvalues of $H$.
$\sum_{i,j} 
\delta_{ij} (-1)^{i+j} | E_k~ Id_N - H |_{ij} $
in the second term is $P'(E_k)$.
Then we find 
\begin{align}
\sum_{i,j} H_{ij} \frac{\partial  Z(E, \eta)}{\partial H_{ij}}
= \sum_k E_k  \frac{\partial  Z(E, \eta)}{\partial E_{k}} .
\label{H-E_formula2}
\end{align}
Next step, we rewrite the Laplacian 
$\displaystyle \sum_{i,k}
\left( 
\frac{\partial  }{\partial H_{ki}}\frac{\partial  }{\partial H_{ik}}
\right)Z(E, \eta)$
by $E_k$.
It is a well-known fact that
 by using the
Vandermonde  determinant 
$\Delta (E) := \prod_{k<l} (E_l -E_k)$
the Jacobian for the change of variables is obtained as follows:
\begin{align*}
d H := 
\prod_{i}^N dH_{ii} 
\prod_{k<l} dH_{kl}^{Re}  
dH_{kl}^{Im} 
= \Delta^2 (E) ( \prod_{i}^N dE_{i} ) 
(\prod_{k<l} (U^{-1}dU)_{kl}^{Re} (U^{-1}dU)_{kl}^{Im}) .
\end{align*}
Then the Laplacian is rewritten as
\begin{align}
\sum_{i,k}
\left( 
\frac{\partial  }{\partial H_{ki}}\frac{\partial  }{\partial H_{ik}}
\right)~Z(E, \eta)
&=\frac{1}{\Delta^2 (E)}
\sum_i^N \frac{\partial}{\partial E_i}
\left( \Delta^2 (E) \frac{\partial}{\partial E_i} \right)~ Z(E, \eta)
 \label{H-E_laplace}\\
&=
\left\{
\sum_{i=1}^N \left( \frac{\partial}{\partial E_i} \right)^2
+\sum_{i \neq j} \frac{1}{E_i - E_j}
\left( \frac{\partial}{\partial E_i} - 
\frac{\partial}{\partial E_j} \right)
\right\}~ Z(E, \eta) . \notag
\end{align}
Here $\displaystyle \sum_{i \neq j}$ means 
$\displaystyle \sum_{i,j=1 , i \neq j}^N$.
From (\ref{SD_H}) ,  (\ref{H-E_formula2}) , and (\ref{H-E_laplace}),
we obtain the following.
\begin{theorem}
The partition function defined by (\ref{partitionfunction})
satisfies 
\begin{align}
{\mathcal L}_{SD} Z(E, \eta) = 0 ,
\label{SD_2}
\end{align}
where 
\begin{align}
{\mathcal L}_{SD} :=
\left\{
\frac{\eta}{N} \sum_{i=1}^N \left( \frac{\partial}{\partial E_i} \right)^2
+
\frac{\eta}{N} \sum_{i \neq j} \frac{1}{E_i - E_j}
\left( \frac{\partial}{\partial E_i} - 
\frac{\partial}{\partial E_j} \right)
-2 \sum_k E_k  \frac{\partial }{\partial E_{k}} -N^2
\right\}~. \label{LSD}
\end{align}
\end{theorem}

In Appendix \ref{appenA},
this  Schwinger-Dyson equation is checked 
by using perturbative calculations, as a cross check.

\section{Diagonalization of ${\mathcal L}_{SD}$}\label{sect3}
In this section we prove the main theorem (Theorem \ref{main_thm1}).\\

As the first step we prove the following proposition.
\begin{proposition}\label{prop3_1}
The differential operator ${\mathcal L}_{SD} $
defined in (\ref{LSD}) is transformed into 
the Hamiltonian of the $N$-body harmonic oscillator as
\begin{align}
e^{-\frac{N}{2\eta} \sum_i E_i^2} \Delta(E)
{\mathcal L}_{SD} 
\Delta^{-1}(E) e^{\frac{N}{2\eta} \sum_i E_i^2} 
= \frac{\eta}{N} \sum_{i=1}^N 
\left( 
\frac{\partial}{\partial E_i}
\right)^2  - \frac{N}{\eta}\sum_{i=1}^N (E_i)^2 .
\end{align}
We denote this Hamiltonian by ${\mathcal H}_{HO}$:
\begin{align}\label{hamiltonian_ho}
-{\mathcal H}_{HO}:=  \frac{\eta}{N} \sum_{i=1}^N 
\left( 
\frac{\partial}{\partial E_i}
\right)^2  - \frac{N}{\eta}\sum_{i=1}^N (E_i)^2 .
\end{align}
\end{proposition}
\begin{proof}
The proof is done by direct calculations.\\
We calculate 
$\Delta^{-1}(E) e^{\frac{N}{2\eta} \sum_i E_i^2}  
\frac{\eta}{N} \sum_{i=1}^N 
\left( 
\frac{\partial}{\partial E_i}
\right)^2
 e^{-\frac{N}{2\eta} \sum_i E_i^2} \Delta(E)$ at first.
\begin{align}
&\Delta^{-1}(E) e^{\frac{N}{2\eta} \sum_i E_i^2}  \frac{\eta}{N} \sum_{i=1}^N 
\left( 
\frac{\partial}{\partial E_i}
\right)^2 e^{-\frac{N}{2\eta} \sum_i E_i^2} \Delta(E) \notag \\
=&  \frac{\eta}{N} \sum_{i=1}^N \Big\{
\sum_{j,k=1, j\neq i, k\neq i}^N
\frac{1}{(E_i -E_j )(E_i -E_k)}
-\sum_{j=1, j\neq i}^N \frac{1}{(E_i - E_j)^2} \Big\}  
\label{prop3_1_1}\\
&-\sum_{i=1}^N \Big\{
 \Big(\sum_{j=1, j\neq i}^N \frac{ 2 E_i}{E_i -E_j} \Big)
+ 1 \Big\} \label{prop3_1_2}\\
&+ 
\frac{\eta}{N} \sum_{i=1}^N \left( \frac{\partial}{\partial E_i} \right)^2
+
\frac{\eta}{N} \sum_{i \neq j} \frac{1}{E_i - E_j}
\left( \frac{\partial}{\partial E_i} - 
\frac{\partial}{\partial E_j} \right)
-2 \sum_{k=1}^N E_k  \frac{\partial }{\partial E_{k}} 
+ \frac{N}{\eta} \sum_{i=1}^N (E_i)^2.
\end{align}
(\ref{prop3_1_1}) is equal to $0$ since
\begin{align*}
&\sum_{j\neq i, k\neq i, j\neq k}
\frac{1}{(E_i -E_j )(E_i -E_k)} \\
&= \frac{1}{3} \sum_{j\neq i, k\neq i, j\neq k}
\Big(
\frac{1}{(E_i -E_j )(E_i -E_k)}+
\frac{1}{(E_j -E_i )(E_j -E_k)}+
\frac{1}{(E_k -E_i )(E_k -E_j)}
\Big)=0.
\end{align*}
(\ref{prop3_1_2}) is written as follows.
\begin{align}
-\sum_{i=1}^N \Big\{
 \Big(\sum_{j=1, j\neq i}^N \frac{ 2 E_i}{E_i -E_j} \Big)
+ 1 \Big\} 
&= - \sum_{i\neq j} 
 \Big( \frac{ E_i}{E_i -E_j} - 
 \frac{ E_j}{E_i -E_j} \Big)
- N
\notag \\
& = - (N^2 -N )-N = -N^2 .
\end{align}
Then we obtain
\begin{align}
\Delta^{-1}(E) e^{\frac{N}{2\eta} \sum_i E_i^2} ~{\mathcal H}_{HO} ~e^{-\frac{N}{2\eta} \sum_i E_i^2} \Delta(E)
= - {\mathcal L}_{SD} .
\end{align}
\end{proof}
We introduce a transformed partition function $\Psi (E, \eta )$ by
\begin{align}\label{Psi_def_3_8}
\Psi (E, \eta ) := 
e^{-\frac{N}{2\eta} \sum_i E_i^2} \Delta(E) Z(E, \eta ) .
\end{align}
Note that this transformation is invertible.
Then the following theorem follows from 
Proposition \ref{prop3_1} immediately.
\begin{theorem} \label{Thm3_2}
The transformed partition function $\Psi (E, \eta ) $
is a zero-energy solution of the 
Schr\"odinger-type differential equation:
\begin{align}
{\mathcal H}_{HO} \Psi(E, \eta ) = 0.
\label{Schrodinger}
\end{align}
Here ${\mathcal H}_{HO}$ is the
$N$-body harmonic oscillator Hamiltonian (\ref{hamiltonian_ho}).
\end{theorem}
This $N$-body harmonic oscillator system has
no interaction terms between the oscillators, so
it is a trivial quantum integrable system.

Theorem \ref{main_thm1} is proved as above. 
In the next section, we calculate the solution $\Psi (E, \eta ) $
more concretely and give another proof using it.

\section{From Partition Function 
to Zero-energy Solution} \label{sectPartitionFun_IZ}
A new expression of the zero-energy solution
of the $N$-body harmonic oscillator system is constructed
by a direct calculation of the partition function.\\

Let us carry out the integration of the
off-diagonal components of $\Phi$ in
the definition of the partition function
after the change of variables to $U(N)\times {\mathbb R}^N$.
We denote the eigenvalues of $\Phi$ by 
$x_1 , x_2 , \cdots , x_N$. 
By using a unitary matrix $U$, $\Phi$ is diagonalized
as $X=U \Phi U^\dagger$, where 
$X= diag(x_1,  x_2 , \cdots , x_N )$.
Then, 
\begin{align}
Z(E, \eta) &= \int_{h_N} d \Phi ~e^{-S'} \notag \\
&= \int_{{\mathbb R}^N} (\prod_{i=1}^N dx_i~  e^{-N V(x_i)}) 
(\prod_{l < k}(x_k -x_l)^2 )
\int_{U(N)} dU e^{-N Tr U E U^{\dagger} X^2 },
\end{align}
where $V(x):= \frac{\eta}{4}x^4$.

Let us use 
the Harish-Chandra-Itzykson-Zuber integral \cite{Itzykson:1979fi,T.Tao} for the unitary group $U(N)$ :
\begin{align}
\int_{U(N)}\exp\left(t\mathrm{tr}\left(AUBU^{\dagger}\right)\right)dU=&\tilde{c}_{N}\frac{\displaystyle\det_{1\leq i,j\leq N}\left(\exp\left(t\lambda_{i}(A)\lambda_{j}(B)\right)\right)}{t^{\frac{(N^{2}-N)}{2}}\displaystyle\Delta(\lambda(A))\displaystyle\Delta(\lambda(B))}.\label{Itzykson}
\end{align}
Here $A$ and $B$ are Hermitian matrices whose eigenvalues are denoted by $\lambda_{i}(A)$ and $\lambda_{i}(B)$ $(i=1,\cdots,N)$, respectively. 
$t$ is a non-zero complex parameter, $\displaystyle\Delta(\lambda(A)):=\prod_{1\leq i<j\leq N}(\lambda_{j}(A)-\lambda_{i}(A))$ is the Vandermonde determinant, and $\displaystyle \tilde{c}_{N}:=\left(\prod_{i=1}^{N-1}i!\right)\times\pi^{\frac{N(N-1)}{2}}$. $\left(\exp\left(t\lambda_{i}(A)\lambda_{j}(B)\right)\right)$ is the $N\times N$ matrix with the $i$-th row and the $j$-th column being $\exp\left(t\lambda_{i}(A)\lambda_{j}(B)\right)$.
After adapting this formula, the partition function is described by
\begin{align}
&Z(E, \eta) = 
\frac{c_N}{\Delta(E)}
\int_{{\mathbb R}^N} (\prod_{i=1}^N dx_i ~ e^{-N V(x_i)}) 
\left(\prod_{l < k}\frac{x_k -x_l}{x_k +x_l} \right)
\det_{1\leq i,j\leq N} \left( e^{-NE_i x_j^2 }
\right) \notag \\
&=\sum_{\sigma \in S_N}
\frac{c_N}{\Delta(E)}
\int_{{\mathbb R}^N} (\prod_{i=1}^N dx_i ~  e^{-N V(x_i)}) 
\left(\prod_{l < k}\frac{x_k -x_l}{x_k +x_l} \right)
 (-1)^\sigma
\prod_{j=1}^N e^{-N E_j x_{\sigma(j)}^2} ,
\label{integral_with_pf}
\end{align}
where $c_N = \tilde{c}_N (-1/ N)^{\frac{N^2-N}{2}}$
and $S_N$ denotes the symmetric group.
This integral representation (\ref{integral_with_pf})
should be regarded as a Cauchy principal value.
Consider the change of variables $x_i \mapsto x_{\sigma^{-1}(i)}$.
Note that the sign of $\prod_{l < k}(x_k -x_l)$ changes
as $(-1)^\sigma \prod_{l < k}(x_k -x_l)$
, and the following formula is obtained by this change of variables.
\begin{align}
Z(E, \eta) &=
\sum_{\sigma \in S_N}
\frac{c_N}{\Delta(E)}
\int_{{\mathbb R}^N} (\prod_{i=1}^N dx_i ~ e^{-N V(x_i)}) 
\left(\prod_{l < k}\frac{x_k -x_l}{x_k +x_l} \right)
\prod_{j=1}^N e^{-N E_j x_{j}^2} \notag \\
 &=\frac{N! c_N}{\Delta(E)}
\int_{{\mathbb R}^N} \left(\prod_{i=1}^N dx_i ~  e^{-N (V(x_i) + E_i x_i^2)}
\right) 
\left(\prod_{l < k}\frac{x_k -x_l}{x_k +x_l} \right).
\label{partition_function_1}
\end{align}

Then the zero-energy solution of (\ref{Schrodinger}) is obtained by
(\ref{correspPsi_Z}).
\begin{theorem}\label{thm4_1}
The function
\begin{align}
\Psi(E, \eta) =N! c_N 
\int_{{\mathbb R}^N} \left(\prod_{i=1}^N dx_i ~  
e^{-N (\frac{\eta}{4} x_i^4 + E_i x_i^2 + \frac{1}{2\eta} E_i^2)}
\right) 
\left(\prod_{l < k}\frac{x_k -x_l}{x_k +x_l} \right) 
\label{Psi_phi4}
\end{align}
satisfies the Schr\"odinger-type differential equation
(\ref{Schrodinger}).
\end{theorem}
Since this fact follows from Theorem \ref{Thm3_2}, there is no need to prove it, but it would be worthwhile to show the differential equation (\ref{Schrodinger}) directly from expression 
(\ref{Psi_phi4}) as a confirmation.\\

At first, we prove the following Lemma:
\begin{lemma}\label{lemma3}
\begin{align} \label{Lemma3_eq}
\int_{{\mathbb R}^N} \left(\prod_{i=1}^N dx_i ~  
e^{-N (\frac{\eta}{4} x_i^4 + E_i x_i^2 + \frac{1}{2\eta} E_i^2)}
\right) 
\left(\prod_{l < k}\frac{x_k -x_l}{x_k +x_l} \right)
\sum_{j=1}^N
(\eta N x_j^4 +2N E_j x_j^2 -1 )
= 0 .
\end{align}
\end{lemma}

\begin{proof}
For simplicity, 
$\displaystyle \sum_i^N
N (\frac{\eta}{4} x_i^4 + E_i x_i^2 + \frac{1}{2\eta} E_i^2)$ will be abbreviated as $f(x,E)$. From the following identity:
\begin{align}
& \int_{{\mathbb R}^N} (\prod_{i=1}^N dx_i ) 
~ \frac{d}{dx_j} 
\Big\{
x_j
e^{-f(x,E)} 
(\prod_{l < k}(x_k -x_l)^2 )
\int_{U(N)} dU e^{-N Tr U E U^{\dagger} X^2 }
\Big\} =0,
\end{align}
we obtain the formal expression:
\begin{align*}
 \sum_{j=1}^N
\int_{{\mathbb R}^N} \left(\prod_{i=1}^N dx_i  \right)
\frac{d}{d x_j}\left\{ x_j
e^{-f(x,E)}
\left(\prod_{l < k}\frac{x_k -x_l}{x_k +x_l} \right)
\right\}=0.
\end{align*}
Then we get
\begin{align*}
\int_{{\mathbb R}^N} \left(\prod_{i=1}^N dx_i  \right) 
&e^{-f(x,E)}
\left(\prod_{l < k}\frac{x_k -x_l}{x_k +x_l} \right) \notag\\
&\times 
\left\{ \sum_{j=1}^N
(\eta N x_j^4 +2N E_j x_j^2 -1 )
-\sum_{m \neq n} \frac{2x_m x_n}{(x_m -x_n)(x_m +x_n)} 
\right\}
= 0.
\end{align*}
From the identity
\begin{align*}
\sum_{m \neq n} \frac{2x_m x_n}{(x_m -x_n)(x_m +x_n)} 
=\sum_{m \neq n} \Big( 
\frac{x_m x_n}{(x_m -x_n)(x_m +x_n)} 
+\frac{x_n x_m}{(x_n -x_m)(x_n +x_m)}
\Big)
=0 ,
\end{align*}
(\ref{Lemma3_eq}) is obtained.
\end{proof}

Using Lemma \ref{lemma3},
let us prove Theorem \ref{thm4_1} by direct
calculations.
\begin{proof}
From (\ref{Psi_phi4}),
\begin{align}
&\frac{\eta}{N} \sum_i^N \left(
\frac{\partial}{\partial E_i} \right)^2
\Psi (E, \eta ) \notag \\
&=N! c_N 
\int_{{\mathbb R}^N} \left(\prod_{j=1}^N dx_j  \right) 
e^{-f(x,E)}
\left(\prod_{l < k}\frac{x_k -x_l}{x_k +x_l} \right)
\left\{ \sum_{i=1}^N
(\eta N x_i^4 +2N E_i x_i^2 -1 + \frac{N}{\eta}E_i^2) \right\}
\notag \\
&= \frac{N}{\eta} \sum_i^N E_i^2 \Psi (E, \eta ) .
\end{align}
For the last equality, we used Lemma \ref{lemma3}.
\end{proof}
$\Psi (E, \eta )$ can also be expressed using Pfaffian. 
It is described in Appendix \ref{apenB}.\\
\bigskip

As described above, we have also directly proved that 
the function $\Psi (E, \eta )$ obtained from 
the partition function of the matrix model 
satisfies the Schr\"odinger-type differential equation 
for the $N$-body harmonic oscillator 
system without  interactions.\\

\section{Discussions and Remarks for $N=1$}\label{appenC}
The matrix model studied in this paper is related to a
renormalizable scalar $\Phi^4$ theory on Moyal space \cite{Grosse:2012uv}
in the large $N$ limit.
There are mainly two approaches to study the question of integrability of this matrix model:
One relies on the model, where one replaces the $\Phi^4$ interaction by a constant times $\Phi^3$. 
This gives the Kontsevich model, for which it is known, that the logarithm of the partition function 
is the $\tau$ function for the KdV hierarchy and fulfills a Hirota bilinear equation\cite{Witten:1990hr,Kontsevich:1992ti,Itzykson:1992ya,Harnad2021}. 
Another approach 
follows topological recursion. 
While the Kontsevich model follows topological recursion, 
it turned out, that the $\Phi^4$ model follows the more sophisticated blobbed topological recursion,
(proven for genus one and two). \cite{Branahl:2021slr,Branahl:2020yru,Hock:2021tbl}. 
Due to these complications, it was unexpected, 
to obtain such a simple answer. 
This $N$-body harmonic oscillator system is known as an integrable system
and this system has been studied for a long time.
See for example \cite{Olshanetsky,Semay:2011zz,Willemyns:2021rgm} and references therein.
Note that the solution required by the Schr\"odinger-type equation 
(\ref{Schrodinger}) is a zero-energy solution, which is different 
from the well-known harmonic oscillator solutions by using 
Hermite polynomials for non-zero energy solutions.
In particular, the case $N=1$ corresponds to what is called the Weber equation.
In the following, we will consider the case $N=1$ as a particularly simplest case 
and see how $\Psi (E, \eta)$ corresponds to a solution to the Weber equation.\\

Introducing new variables $\displaystyle u_i := \sqrt{\frac{N}{\eta}}E_i$,
the Schr\"odinger-type equation (\ref{Schrodinger}) is deformed into
\begin{align}
\sum_{i=1}^N \left( \frac{\partial}{\partial u_i} \right)^2 y(u) =  \sum_{i=1}^N u_i^2 y(u) .
\end{align}
So the $N=1$ case, this is a kind of the Weber equation\cite{DARWIN}:
\begin{align}
 y'' (u) =  u^2 y(u) . \label{weber}
\end{align}
The series solution of this Weber equation is given as follows.
For $y(u)= \sum_{n=0}^{\infty} a_n u^n$, (\ref{weber}) requires
\begin{align}
&a_{4n+2}=a_{4n+3}=0, \notag \\
&a_{4n} = \frac{1}{4n(4n-1)}\cdot \frac{1}{(4n-4)(4n-5)} \cdots \frac{1}{4\cdot 3} a_0 \\
&a_{4n+1} = \frac{1}{(4n+1)4n}\cdot \frac{1}{(4n-3)(4n-4)} \cdots \frac{1}{5\cdot 4} a_1 .
\notag
\end{align}
So the boundary conditions are given by $y(0)=a_0$ and $y'(0) =a_1$.
For $N=1$ case, the partition function is
\begin{align}
Z(E, \eta):=& \int_{-\infty}^{\infty} dx ~e^{-Ex^2 - \frac{\eta}{4}x^4}
=\int_{-\infty}^{\infty} dx ~e^{-\sqrt{\eta} u x^2 - \frac{\eta}{4} x^4}=: Z(u,\eta),
\end{align}
and using this $Z(u,\eta),$ (\ref{Psi_def_3_8}) implies that the solution of (\ref{weber}) is given as
\begin{align}
\Psi(u) := e^{-\frac{E^2}{2\eta}} Z(u,\eta) =   e^{-\frac{u^2}{2}} \int_{-\infty}^{\infty} dx~ e^{-\sqrt{\eta} u x^2 - \frac{\eta}{4} x^4} .
\end{align}
Indeed, we can prove that $\Psi(u)$ satisfies (\ref{weber}) as follows.
As similar to the proof for 
Lemma \ref{lemma3}, 
$$0= \int_{-\infty}^{\infty} dx \frac{d}{dx}( x e^{-f(u)}) =
\int_{-\infty}^{\infty} dx (1- 2\sqrt{\eta} u x^2 - \eta x^4 ) e^{-f(u)},$$
where $-f(u)= -\sqrt{\eta} u x^2 - \frac{\eta}{4} x^4 -\frac{u^2}{2}$.
Using this formula, (\ref{weber}) is derived:
\begin{align*}
\left( \frac{d}{d u} \right)^2 \Psi(u) 
= \int_{-\infty}^{\infty} dx (u^2 + 2\sqrt{\eta} u x^2 + \eta x^4 -1 )
e^{-f(u)} = u^2 \Psi(u) .
\end{align*}
Furthermore, for $u>0$ , 
by using modified Bessel function of the second kind $K_{\frac{1}{4}}(u)$,
$\Psi(u) $ can also be written as
\begin{align*}
\Psi(u) = \frac{1}{\eta^{\frac{1}{4}}} \sqrt{u} K_{\frac{1}{4}}
\big( \frac{u^2}{2} \big).
\end{align*}
We find that the boundary condition for this solution is
required as
\begin{align*}
a_0 &= \Psi(0) =Z(0,\eta) = \int_{-\infty}^{\infty} dx e^{- \frac{\eta}{4} x^4}
= \frac{\Gamma (\frac{1}{4})}{\sqrt{2}\eta^{\frac{1}{4}} } ,\\
a_1 &= \Psi' (0) =-\sqrt{\eta} \int_{-\infty}^{\infty} dx~ x^2 e^{- \frac{\eta}{4} x^4}
=- \frac{\sqrt{2} \Gamma (\frac{3}{4})}{\eta^{\frac{1}{4}} }.
\end{align*}
Thus, in the case of $N=1$, the results are derived using known special functions.

\bigskip

{\bf Acknowledgement}\\
{
Authors are grateful to Raimar Wulkenhaar and Naoyuki Kanomata for lots of meaningful discussions at ESI.
We also thank Kenji Iohara, Taro Kimura, and Ryuichi Nakayama 
for sharing various techniques and information with us.
A.S. was supported by JSPS KAKENHI Grant Number 21K03258. This study was supported by Erwin Schr\"odinger International Institute for Mathematics and Physics (ESI) through the project ``Research in Teams Project Integrability".
}

\appendix

\section{Appendix : Perturbative Check for Schwinger-Dyson Equation}
\label{appenA}
In Appendix \ref{appenA},  we check that the partition function given in Section \ref{sectPartitionFun_IZ}
satisfies the Schwinger-Dyson equation (\ref{SD_2})
by using 
perturbation theory.\\

For simplicity we use the action (\ref{act1}) in the following perturbative calculations. 
We use the fact that the theory itself is equivalent even if 
the action is changed from (\ref{action_S}) to (\ref{act1}).
Then the partition function is written as
$$
Z(E, \eta)= \int_{h_N} d \Phi ~e^{-S'} .
$$
Note that
\begin{align}
\frac{\eta}{N} \sum_{i=1}^N
\left(
\frac{\partial  }{\partial E_{i}} 
\right)^2 Z(E, \eta) &= \eta N 
\sum_{i,j,k}^N \langle \Phi_{ij}\Phi_{ji}
\Phi_{ik}\Phi_{ki} \rangle , \label{maru1}\\
\frac{\eta}{N} \sum_{i\neq j}
\frac{1}{E_i -E_j}\left(
\frac{\partial  }{\partial E_{i}} -\frac{\partial  }{\partial E_{j}}
\right) Z(E, \eta) &=
-\eta \sum_{i\neq j} \frac{1}{E_i -E_j}
\sum_{k=1}^N \left(
\langle  \Phi_{ik}\Phi_{ki} \rangle - \langle  \Phi_{jk}\Phi_{kj} \rangle
\right)    , \label{maru2} \\
-2 \sum_k E_k \frac{\partial  }{\partial E_{k}} 
Z(E, \eta) &=
2N  \sum_{k,j} E_k \langle  \Phi_{kj}\Phi_{jk} \rangle .
\label{maru3}
\end{align}
The Schwinger-Dyson equation (\ref{SD_2}) is expressed as
\begin{align}
(\ref{maru1})+(\ref{maru2})+(\ref{maru3}) = N^2 Z(E, \eta) .
\label{SD_p}
\end{align}
So, we verify this equation (\ref{SD_p}) by perturbation
expansion
up to first order in $\eta$.
We calculate the three 
expectation values 
(\ref{maru1}), (\ref{maru2}), and (\ref{maru3})
by using Feynman rules.\\

We shall now summarize the Feynman rules.
We denote the free model action by
\begin{align}
S_f&= N~ Tr E \Phi^2 
= N ( 
\sum_{i,j}^N  E_{i}\Phi_{ij}\Phi_{ji} ).
\end{align}
Introducing an $N\times N$ Hermitian matrix $J$
as an external field, the partition function for the
free theory is defined as
\begin{align}
Z_f [J] := \int_{h_N} d\Phi ~
e^{-N  Tr E \Phi^2 +Tr J\Phi } .
\end{align}
After Gaussian integral, this is expressed as
\begin{align}
Z_f [J] := Z_f e^{\frac{1}{2N} \sum_{i,j} J_{ij} J_{ji} \frac{1}{E_i + E_j}}, 
\label{FreePart}
\end{align}
where 
\begin{align}
Z_f := c \frac{1}{\prod_i (E_i)^{\frac{1}{2}} \cdot \prod_{l>k}(E_l +E_k)} ,
\quad  ( c= (\sqrt{2\pi /N})^{N^2} ).
\end{align}
We denote the vacuum expectation value of the free action $S_f$
by $\langle O \rangle_f := \int_{h_N} d\Phi ~ O e^{-S_f}$.
Note that we do not normalize this vacuum expectation value,
i.e. $\langle 1 \rangle_f = Z_f $.
From (\ref{FreePart}), the propagator is given by
\begin{align}
\langle \Phi_{ab}\Phi_{cd} \rangle_f 
= Z_f \frac{1}{N} \frac{\delta_{ad} \delta_{bc}}{E_a + E_b}
=: Z_f~ \lower1.3ex\hbox{\propagator} ,
\end{align}
the interaction is given as 
\begin{align}
 \lower1.9ex\hbox{\vertexs}:= -\frac{\eta}{4} \times 
 \mbox{symmetry factor},
\end{align}
and each loop corresponds to a sum $\sum_n$.
The Feynman diagrams contributing to the $2$-point function 
$\langle \Phi_{ij}\Phi_{kl} \rangle$
to zeroth and first order in $\eta$ are shown in Figure \ref{Feynman}.
\begin{figure}[h]
\begin{center}
\includegraphics[width=0.6 \textwidth]{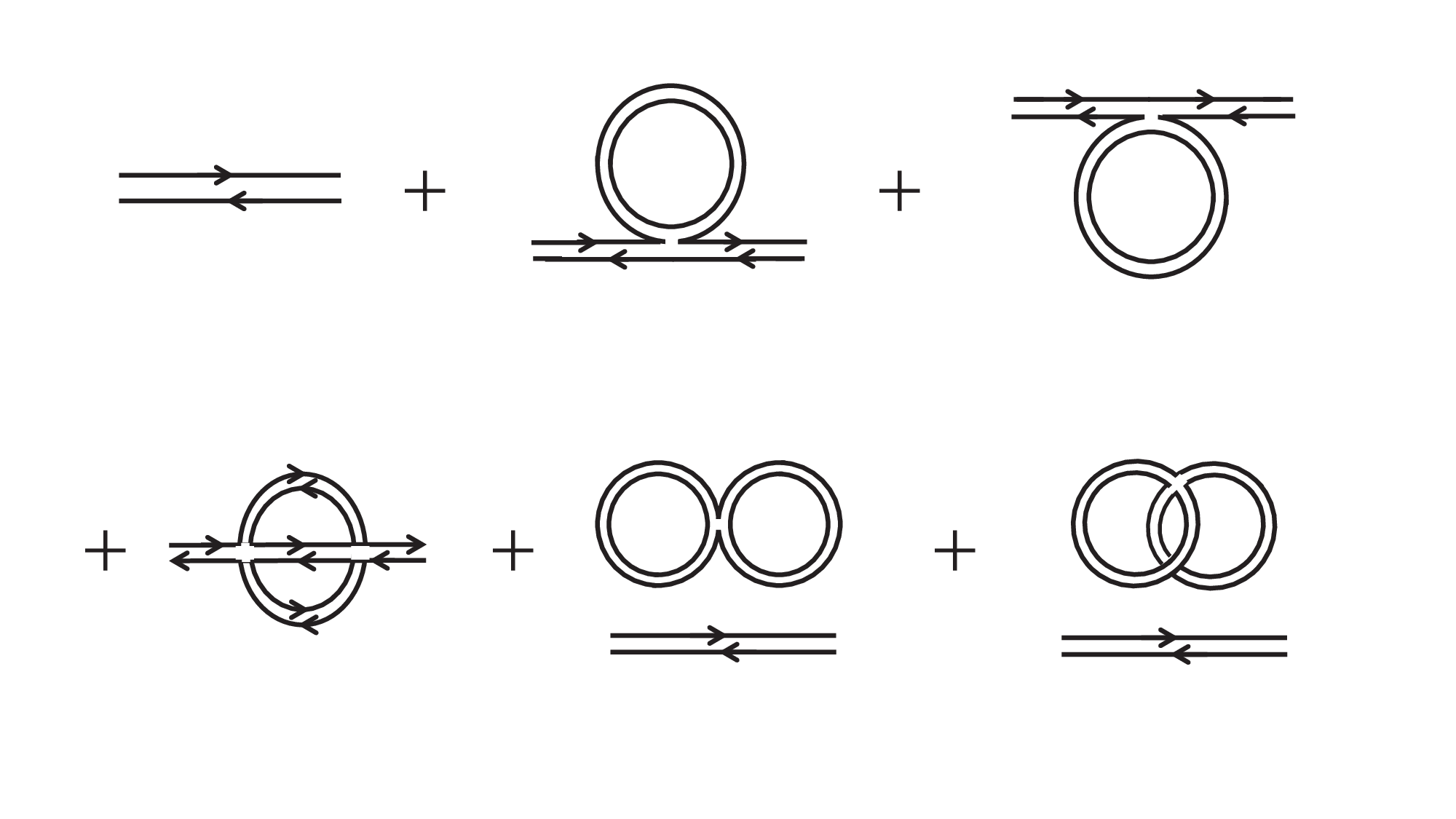}
\caption{Feynman diagrams for $2$-point function}
\label{Feynman}
\end{center}
\end{figure}
If the terms are written in the same order as each graph in Figure \ref{Feynman}, they contribute as follows.
\begin{align}
&\langle \Phi_{ij}\Phi_{kl} \rangle = 
 Z_f\Big\{
\frac{1}{N}\frac{\delta_{kj} \delta_{il}}{E_i + E_j} \notag \\
&+\frac{1}{N^3}\frac{\delta_{kj} \delta_{il}}{(E_i + E_j)^2}
\sum_n \frac{1}{E_i + E_n}
\frac{-\eta N}{4}\times 4
+\frac{1}{N^3}\frac{\delta_{kj} \delta_{il}}{(E_i + E_j)^2}
\sum_n \frac{1}{E_j + E_n}
\frac{-\eta N}{4}\times 4 \notag \\
&+\frac{1}{N^3}\frac{\delta_{ij} \delta_{kl}}{2E_i 2E_k(E_i + E_k)}
\frac{-\eta N}{4}\times 4
+\frac{1}{N}\frac{\delta_{kj} \delta_{il}}{E_i + E_j} \times B \Big\}
+O(\eta^2),
\end{align}
where $B$ represents the terms from bubble graphs in Figure \ref{Feynman}:
\begin{align}
B := -\frac{1}{N}\eta \Big(
\frac{1}{2}\sum_{m,n,s} \frac{1}{(E_m + E_n)(E_m + E_s)}
+\frac{1}{16} \sum_m \frac{1}{E_m^2}
\Big).
\end{align}
(See also (\ref{NextOrderZ}).)
This equation can be rearranged as follows.
\begin{align}
\langle \Phi_{ij}\Phi_{kl} \rangle =  
& \frac{Z_f}{N}\frac{\delta_{kj} \delta_{il}}{E_i + E_j}\Big\{
(1+B) 
- \frac{\eta}{N}\frac{1}{(E_i + E_j)}
\sum_n \left( \frac{1}{E_i + E_n} + \frac{1}{E_j + E_n} \right) \Big\}
\notag\\
&+\frac{Z_f}{N^2}
\frac{ -\eta \delta_{ij} \delta_{kl}}{4E_i E_k(E_i + E_k)} +O(\eta^2) .
\label{2pointfunction_general}
\end{align}

Now that we are ready to calculate contributuions by using Feynman rules, we shall immediately carry out
 the calculations of the three expectation values 
(\ref{maru1}), (\ref{maru2}), and (\ref{maru3}).
At first, we estimate (\ref{maru1}).
By using Wick contraction, we get
\begin{align}
&\eta N 
\sum_{i,j,k}^N \langle \Phi_{ij}\Phi_{ji}
\Phi_{ik}\Phi_{ki} \rangle 
 = \notag \\
& \eta \frac{N}{Z_f} 
\sum_{i,j,k}^N \{ \langle \Phi_{ij}\Phi_{ji}\rangle 
\langle \Phi_{ik}\Phi_{ki}\rangle 
+ \langle \Phi_{ij}\Phi_{ik}\rangle 
\langle \Phi_{ji}\Phi_{ki}\rangle
+\langle \Phi_{ij}\Phi_{ki}\rangle 
\langle \Phi_{ji}\Phi_{ik}\rangle  \}. \label{maru1_2}
\end{align}
Substituting (\ref{2pointfunction_general}) for (\ref{maru1_2}),
finally we obtain
\begin{align}
\eta N 
\sum_{i,j,k}^N \langle \Phi_{ij}\Phi_{ji}
\Phi_{ik}\Phi_{ki} \rangle 
 = \frac{\eta Z_f  }{N}
 \Big\{
 \sum_{i,j,k}^N
  \frac{1}{ (E_i + E_j)(E_i + E_k)}
 + \sum_{i=1}^N 
  \frac{1}{4 E_i^2 }
+ \sum_{i, j}^N 
  \frac{1}{ (E_i +E_j)^2 }
\Big\} +O(\eta^2 ) . \label{maru1_3}
\end{align}
Next, after substituting (\ref{2pointfunction_general}),
(\ref{maru2}) is given as
\begin{align}
-\eta \sum_{i\neq j} \frac{1}{E_i -E_j}
\sum_{k=1}^N \left(
\langle  \Phi_{ik}\Phi_{ki} \rangle - \langle  \Phi_{jk}\Phi_{kj} \rangle
\right)  = 
\frac{\eta Z_f  }{N}
\sum_{i \neq j} \sum_{k=1}^N 
  \frac{1}{ (E_i + E_k)(E_j + E_k)}
+O(\eta^2 ) . \label{maru2_2}
\end{align}
Similarly, (\ref{maru3}) is given as
\begin{align}
2N  \sum_{k,j} E_k \langle  \Phi_{kj}\Phi_{jk} \rangle 
=& 2Z_f \sum_{k, j}
\frac{E_k}{ E_k + E_j}
\left\{
1+B - \frac{\eta}{N} \frac{1}{ E_k + E_j}
\sum_{n=1}^N \Big(
 \frac{1}{ E_k + E_n} +  \frac{1}{ E_j + E_n}
\Big)
\right\} \notag \\
& 
- \frac{\eta Z_f}{4 N} \sum_{k=1}^N \frac{1}{ E_k^2}  +O(\eta^2 ) .
\label{maru3_2}
\end{align}
\bigskip

For the right hand side of  (\ref{SD_p}), 
next we calculate 
\begin{align}
N^2 Z(E, \eta) =
N^2 \sum_{k=0}^{\infty} \frac{1}{k!}
\Big( -\frac{ N \eta}{4} \Big)^k
\langle (Tr \Phi^4 )^k \rangle_f .
\end{align}
$\eta^0$-term is $N^2 \langle 1 \rangle_f = N^2 Z_f$.
$\eta^1$-term is given by 
$ -\frac{ N^3 \eta}{4} 
\langle Tr \Phi^4 \rangle_f $.
This is the same calculation for $B$:
\begin{align}
&-\frac{ N^3 \eta}{4} 
\langle Tr \Phi^4 \rangle_f \notag \\
&= -\frac{ N^3 \eta}{4 Z_f} 
\sum_{i,j,k,l}^N \{ \langle \Phi_{ij}\Phi_{jk}\rangle_f 
\langle \Phi_{kl}\Phi_{li}\rangle_f 
+ \langle \Phi_{ij}\Phi_{kl}\rangle_f 
\langle \Phi_{jk}\Phi_{li}\rangle_f
+\langle \Phi_{ij}\Phi_{li}\rangle_f 
\langle \Phi_{jk}\Phi_{kl}\rangle_f  \}. 
\notag \\
&=-\frac{ N \eta Z_f}{4} 
\sum_{i,j,k,l}^N
\left(
 \frac{\delta_{ik}}{ (E_i + E_j)(E_k + E_l)}
 +\frac{\delta_{jk} \delta_{il} \delta_{kl} \delta_{ji} }{ (E_i + E_j)(E_j + E_k)}
 + \frac{\delta_{jl}}{ (E_i + E_j)(E_j + E_k)}
\right) \notag \\
&= N^2  Z_f B
\label{NextOrderZ}
\end{align}

\bigskip
Now is the time to check the Schwinger-Dyson equation (\ref{SD_p}).
(\ref{maru1}), (\ref{maru2}), and (\ref{maru3})
are calculated as 
(\ref{maru1_3}), (\ref{maru2_2}), and (\ref{maru3_2}),
respectively.
Note that $B$ is the term proportional to $\eta$.
The $\eta^0$-term in
$(\ref{maru1})+ (\ref{maru2})+ (\ref{maru3}) $
exists only in (\ref{maru3_2}):
\begin{align}
2Z_f \sum_{k, j}
\frac{E_k}{ E_k + E_j}=
Z_f \sum_{k, j}\left(
\frac{E_k}{ E_k + E_j}
+\frac{E_j}{ E_k + E_j}
\right) = N^2 Z_f. \label{tric}
\end{align}
On the other hand, $\eta^0$-term from
$N^2 Z(E, \eta)$, that is 
the right hand side of (\ref{SD_p}), 
is $N^2 Z_f$.
Thus, it is confirmed that equation (\ref{SD_p}) 
holds in zeroth order of $\eta$.

Next order terms in 
$(\ref{maru1_3}) + (\ref{maru2_2}) + (\ref{maru3_2})$
are summarized as
\begin{align}
\frac{\eta Z_f  }{N}
 \Big\{
 2 \sum_{i,j,k}^N
  \frac{1}{ (E_i + E_j)(E_i + E_k)}
-2 \sum_{k, j}
\frac{E_k}{ (E_k + E_j)^2}
\sum_{n=1}^N \Big(
 \frac{1}{ E_k + E_n} +  \frac{1}{ E_j + E_n}
\Big)
\Big\}
+ Z_f N^2 B . \label{LHS_1st}
\end{align}
Here we use the same way in (\ref{tric}) 
to obtain the last term.
From the following observation,
\begin{align}
2 \sum_{k, j}
\frac{E_k}{ (E_k + E_j)^2}
\sum_{n=1}^N \Big(
 \frac{1}{ E_k + E_n} +  \frac{1}{ E_j + E_n}
\Big) &=
2 \sum_{i,j,k}
 \frac{1}{(E_i + E_j) (E_i + E_k)} ,
\end{align}
we find (\ref{LHS_1st}) is given by
\begin{align}
Z_f N^2 B .
\end{align}
On the other hand, $\eta$-linear terms from
$N^2 Z(E, \eta)$ are given by 
(\ref{NextOrderZ}),
then we checked that equation (\ref{SD_p}) 
holds in first order of $\eta$.

\section{Appendix : Pfaffian Expression}\label{apenB}
We can rewrite $Z(E, \eta)$ and $\Psi(E, \eta)$
  by using Pfaffian.
For simplicity, we consider only $2N$ case of (\ref{partition_function_1}).
\begin{align}
Z(E, \eta) 
 &=\frac{(2N)! c_{2N}}{\Delta(E)}
\int_{{\mathbb R}^{2N}} \left(\prod_{i=1}^{2N} dx_i  
e^{-2N (V(x_i) + E_i x_i^2)}
\right) 
\left(\prod_{l < k}\frac{x_k -x_l}{x_k +x_l} \right)
\\
&= \frac{(2N)! c_{2N}}{\Delta(E)}
\int_{{\mathbb R}^{2N}}
\left(\prod_{k=1}^{2N} dx_k  
\right) 
\det_{i,j} \left(\delta_{ij} e^{-2N (V(x_i) + E_i x_i^2)} \right)
\text{pf}_{k,l} \left(\frac{x_k -x_l}{x_k +x_l} \right)
. \notag 
\end{align}
As noted in Section \ref{sectPartitionFun_IZ},
the integral shall be treated as a principal value integral.
We use $C= ((2N)!)^2 c_{2N}$ below.
After applying Bruijn's formula \cite{Bruijn:1995}, this is obtained
by
\begin{align}
Z(E, \eta) 
 &= \frac{C}{\Delta(E)}
 \text{pf}_{i,j}
 \langle e^{-2N(V(x_i)+E_i x_i^2)} \Big| 
 \left(\frac{x_k -x_l}{x_k +x_l} \right)
 \Big| e^{-2N(V(x_j)+E_j x_j^2)} \rangle ,
 \end{align}
where 
\begin{align}
& \langle e^{-2N(V(x_i)+E_i x_i^2)} \Big| 
 \left(\frac{x_k -x_l}{x_k +x_l} \right)
 \Big| e^{-2N(V(x_j)+E_j x_j^2)} \rangle \notag \\
 &:= \int_{{\mathbb R}^{2}} dx_i dx_j
 \left( 
 \frac{x_i -x_j}{x_i +x_j}\right)
 e^{-2N(V(x_i)+E_i x_i^2 + V(x_j)+E_j x_j^2)} .
\end{align}
The justification for this process is discussed in \cite{Borot:2023thu}.
Therefore the zero-energy solution
is given by
\begin{align}
\Psi(E, \eta) 
 &= {C}
 e^{-\frac{N}{2\eta} \sum_i E_i^2 }
 \text{pf}_{i,j}
 \langle e^{-2N(V(x_i)+E_i x_i^2)} \Big| 
 \left(\frac{x_k -x_l}{x_k +x_l} \right)
 \Big| e^{-2N(V(x_j)+E_j x_j^2)} \rangle .
\end{align}



\end{document}